\documentclass{sig-alternate-2013}

\usepackage{color,array}
\usepackage{epstopdf}
\usepackage{paralist}

\usepackage{dgleich-setup}

\usepackage[utf8]{inputenc}

\usepackage{amsfonts,amsmath,amssymb}

\usepackage{array}

\newcommand{\lrb}[1]{\left[ #1 \right]}

\newcommand{\lp}{\left(}
\newcommand{\rp}{\right)}

\providecommand{\kron}{\otimes}

\newcommand{\bmat}[1]{\lrb{ \begin{array}{c} #1 \end{array} }}

\providecommand{\eye}{\textbf{I}}
\providecommand{\mA}{\ensuremath{\textbf{A}}}

\providecommand{\mD}{\ensuremath{\textbf{D}}}

\providecommand{\mG}{\ensuremath{\textbf{G}}}

\providecommand{\mI}{\ensuremath{\textbf{I}}}

\providecommand{\mP}{\ensuremath{\textbf{P}}}

\providecommand{\mS}{\ensuremath{\textbf{S}}}

\providecommand{\mX}{\ensuremath{\textbf{X}}}

\providecommand{\ve}{\ensuremath{\textbf{e}}}
\providecommand{\vf}{\ensuremath{\textbf{f}}}

\providecommand{\vh}{\ensuremath{\textbf{h}}}

\providecommand{\vp}{\ensuremath{\textbf{p}}}

\providecommand{\vr}{\ensuremath{\textbf{r}}}
\providecommand{\vs}{\ensuremath{\textbf{s}}}

\providecommand{\vv}{\ensuremath{\textbf{v}}}

\providecommand{\vx}{\ensuremath{\textbf{x}}}
\providecommand{\vy}{\ensuremath{\textbf{y}}}
\providecommand{\vz}{\ensuremath{\textbf{z}}}

\newcommand{\iv}{^{-1} }

\RequirePackage{graphicx}
\RequirePackage{tabularx}
\RequirePackage{booktabs}
\RequirePackage{ragged2e}
\RequirePackage{xcolor}
\RequirePackage{textcomp} 
\RequirePackage{xspace}
\RequirePackage{microtype}
\RequirePackage{listings}

\colorlet{TufteRed}{red!80!black}
\definecolor{halfgray}{gray}{0.55}

\definecolor{subtleblue}     {rgb}{0.02,0.04,0.48}
\definecolor{subtlered}      {rgb}{0.65,0.04,0.07} % RGB 165,10,18
\definecolor{subtlegreen}    {rgb}{0.06,0.44,0.08}
\definecolor{subtledarkblue} {rgb}{0,.1,.6}
\definecolor{lightsubtleblue}{rgb}{0,.4,.6}
\definecolor{ecru}           {rgb}{1.0,.98823,.95686}   

\definecolor{stanfordred}      {rgb}{0.6431,0.000,0.1137} 

\definecolor{stanfordsandstone}{rgb}{0.9059,0.8196,0.6039}
\definecolor{stanfordblue}     {rgb}{0.1451,0.5176,0.7333}
\definecolor{stanfordgreen}    {rgb}{0.1608,0.3961,0.2863} % rgb 41/101/73
\definecolor{stanforddarkgray} {rgb}{0.2627,0.2902,0.2667} % rgb 67/74/68
\definecolor{stanfordlightgray}{rgb}{0.8392,0.8667,0.8275} % rgb 214/221/211
\definecolor{stanforddarkgreen}{rgb}{0.2353,0.2118,0.1373}
\definecolor{stanforddeepred}  {rgb}{0.6510,0.2275,0.0000}
\definecolor{stanfordneutralkhaki}{rgb}{0.5686,    0.5333,    0.4510}
\definecolor{stanfordbrightgreen}{rgb}{0.0039 ,   0.5137  ,  0.3725}
\definecolor{stanfordbrightblue}{rgb}{0.1451,0.5176, 0.7333} 
\definecolor{stanfordbrightseagreen}{rgb}{0.0000,0.5020,0.5529} % rgb 0/128/141
\definecolor{stanfordbrightyellow}{rgb}{0.9412,0.6863,0.0000} % rgb 240/175/0
\definecolor{stanfordbrightwine}{rgb}{0.2353,0.0667,0.0275}

\newcommand{\PreserveBackslash}[1]{\let\temp=\\#1\let\\=\temp}
\newcolumntype{L}[1]{>{\PreserveBackslash\RaggedRight}m{#1}}
\newcolumntype{M}[1]{>{\PreserveBackslash\RaggedRight}p{#1}}
\newcolumntype{R}[1]{>{\PreserveBackslash\RaggedLeft}m{#1}}
\newcolumntype{S}[1]{>{\PreserveBackslash\RaggedLeft}p{#1}}
\newcolumntype{Z}[1]{>{\PreserveBackslash\Centering}m{#1}}
\newcolumntype{A}[1]{>{\PreserveBackslash\Centering}p{#1}}

\newcolumntype{U}{>{\setlength{\RaggedRightParindent}{0pt}\RaggedRight\arraybackslash\noindent}X}
\newcolumntype{V}{>{\RaggedLeft\arraybackslash}X}
\newcolumntype{W}{>{\Centering\arraybackslash}X}

\graphicspath{{.}{./figures/}}

\lstnewenvironment{pseudocode}[1][]{%
  \lstset{language=python,#1%             % use our version of highlighting
  }%
}{}

\lstdefinelanguage{matlabfloz}{%
  alsoletter={...},%
  morekeywords={%                             % keywords
  break,case,catch,continue,elseif,else,end,for,function,global,%
  if,otherwise,persistent,return,switch,try,while,...,ones,zeros,eye},%
  comment=[l]\%,%                             % comments
  morecomment=[l]...,%                        % comments
  morestring=[m]',%                           % strings 
}[keywords,comments,strings]%

\lstnewenvironment{dispmcode}{%
    \lstset{language=matlabfloz,%
      numbers=none,frame=none,%
      keywordstyle=\color{subtleblue},%,          % keywords
      commentstyle=\color{thegreen},%         % comments
      stringstyle=\color{subtleed}%             % strings
    }%
  }%
  {}

\lstnewenvironment{dispmresults}{%
    \lstset{language=matlabfloz,%
      numbers=none,frame=none,backgroundcolor=\color{ecru},%
      keywordstyle=\color{theblue},%,          % keywords
      commentstyle=\color{subtlegreen},%         % comments
      stringstyle=\color{subtlered}%             % strings
    }%
  }%
  {}

\lstnewenvironment{numberedmcode}{%
    \lstset{language=matlabfloz,%
      numbers=left,frame=none,%
      keywordstyle=\color{subtleblue},%,          % keywords
      commentstyle=\color{subtlegreen},%         % comments
      stringstyle=\color{subtlered}%             % strings
    }%
  }%
  {}

\newtheorem{theorem}{Theorem}
\newtheorem{lemma}{Lemma}

\usepackage[colorlinks,citecolor=black,linkcolor=black]{hyperref}

\newcommand{\mLambda}{\boldsymbol{\Lambda}}
\newcommand{\eps}{\varepsilon}
\newcommand{\hkre}{\texttt{hk-relax}\xspace}
\newcommand{\vol}{\text{vol}}
\newcommand{\expof}[1]{\exp\left\{ #1 \right\}}

\newcommand{\tnps}{T_N(t\mP)\vs}

\newcommand{\hvv}{\hat{\vv}}
\newcommand{\fullvec}[1]{\bmat{{#1}_0 \\ \vdots \\ {#1}_N}}

\newfont{\mycrnotice}{ptmr8t at 7pt}
\newfont{\myconfname}{ptmri8t at 7pt}
\let\crnotice\mycrnotice%
\let\confname\myconfname%

\permission{Permission to make digital or hard copies of all or part of this work for personal or classroom use is granted without fee provided that copies are not made or distributed for profit or commercial advantage and that copies bear this notice and the full citation on the first page. Copyrights for components of this work owned by others than the author(s) must be honored. Abstracting with credit is permitted. To copy otherwise, or republish, to post on servers or to redistribute to lists, requires prior specific permission and/or a fee. Request permissions from permissions@acm.org.}
\conferenceinfo{KDD'14,}{August 24--27, 2014, New York, NY, USA. \\
{\mycrnotice{Copyright is held by the owner/author(s). Publication rights licensed to ACM.}}}
\copyrightetc{ACM \the\acmcopyr}
\crdata{978-1-4503-2956-9/14/08\ ...\$15.00.\\
http://dx.doi.org/10.1145/2623330.2623706}

\begin{document}

\title{Heat Kernel Based Community Detection}

\numberofauthors{2}
\author{
% 1st. author
\alignauthor
Kyle Kloster\\
       \affaddr{Purdue University}\\
       \affaddr{West Lafayette, IN}\\
       \email{kkloste@purdue.edu}
% 2nd. author
\alignauthor
David F. Gleich\\
       \affaddr{Purdue University}\\
       \affaddr{West Lafayette, IN}\\
       \email{dgleich@purdue.edu}
}
\date{12 February 2014\footnote{Minor correction posted 14 November 2016}}

\maketitle
\begin{abstract}
The heat kernel is a type of graph diffusion that, like the much-used personalized PageRank diffusion, is useful in identifying a community nearby a starting seed node. We present the first deterministic, local algorithm to compute this diffusion and use that algorithm to study the communities that it produces. Our algorithm is formally a relaxation method for solving a linear system to estimate the matrix exponential in a degree-weighted norm. We prove that this algorithm stays localized in a large graph and has a worst-case constant runtime that depends only on the parameters of the diffusion, not the size of the graph. On large graphs, our experiments indicate that the communities produced by this method have better conductance than those produced by PageRank, although they take slightly longer to compute. On a real-world community identification task, the heat kernel communities perform better than those from the PageRank diffusion.
\end{abstract}

%A category including the fourth, optional field follows...
%\category{D.2.8}{Software Engineering}{Metrics}[complexity measures, performance measures]



\section{Introduction}

The community detection problem is to identify a set of nodes in a graph that are internally cohesive but also separated from the remainder of the network. One popular way to capture this idea is through the conductance measure of a set. We treat this idea formally in the next section, but informally, the conductance of a set is a ratio of the number of edges leaving the set to the number of edges touched by the set of vertices. If this value is small, then it indicates a set with many internal edges and few edges leaving.

In many surveys and empirical studies~\cite{Dhillon2007-graclus,Leskovec-2009-community-structure,Schaeffer-2007-clustering}, the conductance measure surfaces as one of the most reliable measures of a community. Although this measure has been critized for producing \emph{cavemen}-type communities~\cite{Kang-2011-Caveman}, empirical properties of real-world communities correlate highly with sets produced by algorithms that optimize conductance~\cite{Abrahao}. Furthermore, state-of-the-art methods for identifying overlapping sets of communities use conductance to find real-world communities better than any alternative~\cite{Whang-2013-overlapping}.

Virtually all of the rigorous algorithms that identify sets of small conductance are based on min-cuts~\cite{andersen2008-improve,Orecchia-2014-flow}, eigenvector computations~\cite{Alon-1985-cheeger,Fiedler1973-algebraic-connectivity}, or local graph diffusions~\cite{andersen2006-local,chung2007-pagerank-heat}. (One notable exception is the graclus method~\cite{Dhillon2007-graclus} that uses a relationship between kernel $k$-means and a variant of conductance.) In this paper, we study a new algorithm that uses a heat kernel diffusion~\cite{chung2007-pagerank-heat} to identify small-conductance communities in a network. (The heat kernel is discussed formally in Section~\ref{sec:diffusions}.)  Although the properties of this diffusion had been analyzed in theory in Chung's work~\cite{chung2007-pagerank-heat}, that work did not provide an efficient algorithm to compute the diffusion. Recently, Chung and Simpson stated a randomized Monte Carlo method to estimate the diffusion~\cite{chung2013solving}.

This paper introduces an efficient and deterministic method to estimate this diffusion. We use it to study the properties of the small conductance sets identified by the method as communities. For this use, a deterministic approach is critical as we need to differentiate between subtle properties of the diffusion. Our primary point of comparison is the well-known personalized PageRank diffusion~\cite{andersen2006-local}, which has been used to establish new insights into the properties of communities in large scale networks~\cite{Leskovec-2009-community-structure}. Thus, we wish to understand how the communities produced by the heat kernel compare to those produced by personalized PageRank.

The basic operation of our algorithm is a coordinate relaxation step. This has been used to efficiently compute personalized PageRank~\cite{andersen2006-local,jeh2003-personalized,mcsherry2005-uniform} where it is known as the ``push'' operation on a graph; the term ``coordinate relaxation'' is the classical name for this operation, which dates back to the Gauss-Seidel method. What distinguishes our approach from prior work is the use of coordinate relaxation on an implicitly constructed linear system that will estimate the heat kernel diffusion, which is formally the exponential of the random walk transition matrix.

We began looking into this problem recently in a workshop paper~\cite{Kloster-2013-nexpokit}, where we showed that this style of algorithm successfully estimates a related quantity. This paper tackles a fundamentally new direction and, although similar in techniques, required an entirely new analysis. In particular, we are able to estimate this diffusion in constant time in a degree-weighted norm that depends only on the parameters of the diffusion, and not the size of the graph. A Python implementation of our algorithm to accomplish this task is presented in Figure~\ref{fig:pseudocode}.

In the remainder of the paper, we first review these topics formally (Sections~\ref{sec:prelim}~and~\ref{sec:diffusions}); present our algorithm (Section~\ref{sec:alg}) and discuss related work (Section~\ref{sec:related}). Then we show how our new approach differs from and improves upon the personalized PageRank diffusion in synthetic and real-world problems.

\paragraph{Summary of contributions}
\begin{compactitem}[$\cdot$]
 \item We propose the first local, deterministic method to accurately compute a heat kernel diffusion in a graph. The code is simple and scalable to any graph where out-link access is \emph{inexpensive}.
 \item Our method is always localized, even in massive graphs, because it has a provably constant runtime in a degree-weighted norm.
 \item We compare the new heat kernel diffusion method to the venerable PageRank diffusion in synthetic, small, and large networks (up to 2 billion edges) to demonstrate the differences. On large networks such as Twitter, our method tends to produce smaller and tighter communities. It is also more accurate at detecting ground-truth communities.
\end{compactitem}\bigskip

\noindent We make our experimental codes available in the spirit of reproducible research:\\
{\fontsize{8}{10}\selectfont
\url{https://www.cs.purdue.edu/homes/dgleich/codes/hkgrow}}

%$e^t$ vs. $\frac{1}{1-\alpha}$.

% \begin{itemize}
%
% \item[3.6] Brief statement of our WAW results on GS (for power law graph)
% \item[4] Proofs of theorems
%
% \item[6] Experimental Results
% \begin{itemize}
% \item[6.1] Comparison of conductances, cluster sizes
%
% - `smalldata': compare conductance of ppr and hk cross all 12 small graphs. Two side-by-side plots in one figure.
%
% - `heavy degree twitter': see how conductance and setsize vary  for ppr and hk across 2000 nodes from the top $10\%$ largest degree nodes.
%
% \item[6.2] Runtime Comparison
%
% - same as `smalldata' above, with runtimes instead of conductances
%
% - same as `heavy degree twitter' above, but with runtimes.
%
% \end{itemize}
% \item[7] Conclusion
% \end{itemize}

\section{Preliminaries}\label{sec:notation}
\label{sec:prelim}
We begin by fixing our notation. Let $G=(V,E)$ be a simple, undirected graph with $n = |V|$ vertices. Fix an ordering of the vertices from $1$ to $n$ such that we can refer to a vertex by it's numeric ID. For a vertex $v \in V$, we denote by $d_v$ the degree of node $v$. Let $\mA$ be the associated adjacency matrix. Because $G$ is undirected, $\mA$ is symmetric. Furthermore, let $\mD$ be the diagonal matrix of degrees ($\mD_{ii} = d_i$) and $\mP = (\mD\iv\mA)^T = \mA \mD\iv$ be the random walk transition matrix. Finally, we denote by $\ve_i$ the vector (of an appropriate length) of zeros having a single 1 in entry $i$, and by $\ve$ the vector of all 1s.

\textbf{Conductance.} Given a subset $S\subseteq V$, we denote by $\vol(S)$ the sum of the degrees of all vertices in $S$, and by $\partial(S)$, the \emph{boundary} of $S$, the number of edges with one endpoint inside of $S$ and one endpoint outside of $S$. With this notation, the conductance of a set $S$ is given by
\[
\phi(S) := \frac{\partial(S)}{\min\{ \vol(S) , \vol(V - S) \}}
\]
Conceptually, $\phi(S)$ is the probability that a random walk of length one will land outside of $S$, given that we start from a node chosen uniformly at random inside $S$.

\textbf{Matrix exponential.} The heat kernel of a graph involves the matrix exponential, and we wish to briefly mention some facts about this operation. See Higham~\cite{Higham2008-functions-of-matrices} for a more in-depth treatment. Consider a general matrix $\mG$. The exponential function of a matrix is not just the exponential applied element-wise, but rather is given by substituting the matrix into the Taylor series expansion of the exponential function: 
\[ \expof{\mG} = \sum_{k=0}^\infty \frac{1}{k!} \mG^k. \]
That said, the exponential of a diagonal matrix \emph{is} the exponential function applied to the diagonal entries. This phenomenon occurs because powers of a diagonal matrix are simply the diagonal elements raised to the given power. For any matrix $\mG$, if $\mG^T \mG = \mG \mG^T$ (which is a generalization of the notion of a symmetric matrix, called a \textit{normal} matrix), then $\mG$ has an eigenvalue decomposition $\mG = \mX \mLambda \mX\iv$ and there is a simple, albeit inefficient, means of computing the exponential: $\expof{\mG} = \mX \expof{\mLambda} \mX\iv$. Computing the matrix exponential, or even the \emph{product} of the exponential with a vector, $\expof{\mG} \vz$, is still an active area of research~\cite{Al-Mohy-2011-exponential}.

\section{Finding small conductance \\ communities with diffusions}
\label{sec:communities}
\label{sec:diffusions}

A graph diffusion is any sum of the following form:
\[ \vf = \sum_{k=0}^\infty \alpha_k \mP^k \vs \]
where $\sum_{k} \alpha_k = 1$ and $\vs$ is a stochastic vector (that is, it is non-negative and sums to one.) Intuitively, a diffusion captures how a quantity of $s_i$ material on node $i$ \emph{flows} through the graph. The terms $\alpha_k$ provide a decaying weight that ensures that the diffusion eventually dissipates. In the context of this paper, we are interested in the diffusions of single nodes or neighborhood sets of a single vertex; and so in these cases $\vs = \ve_i$ or $\vs = \sum_{i \in S} \ve_i/|S|$ for a small set $S$. We call the origins of the diffusion the \emph{seeds}.

Given an estimate of a diffusion $\vf$ from a seed, we can produce a small conductance community from it using a sweep procedure. This involves computing $\mD^{-1} \vf$, sorting the nodes in descending order by their magnitude in this vector, and computing the conductance of each prefix of the sorted list. Due to the properties of conductance, there is an efficient means of computing all of these conductances. We then return the set of smallest conductance as the community around the seeds.

\textbf{The personalized PageRank diffusion.}
One of the most well-known instances of this framework is the personalized PageRank diffusion. Fix $\alpha \in (0,1)$. Then $\vp$ is defined:
\begin{equation}
\vp = (1-\alpha) \sum_{k=0}^\infty \alpha^k \mP^k \vs.
\end{equation}
The properties of this diffusion have been studied extensively. In particular, Andersen et al.~\cite{andersen2006-local} establish a local Cheeger inequality using a particular algorithm called ``push'' that locally distributes mass. The local Cheeger inequality informally states that, \emph{if} the seed is nearby a set with small conductance, \emph{then} the result of the sweep procedure is a set with a related conductance. Moreover, they show that their ``push'' algorithm estimates $\vf$ with an error $\eps$ in a degree-weighted norm by looking at $\frac{1}{(1-\alpha)\eps}$ edges.

\textbf{The heat kernel diffusion.}
Another instance of the same framework is the heat kernel diffusion~\cite{chung2007-pagerank-heat,chung2009local}. It simply replaces the weights $\alpha_k$ with $t^k/k!$:
\begin{equation}\label{defhk}
\vh = e^{-t}\left( \sum_{k=0}^{\infty} \frac{t^k}{k!} (\mP)^k \right) \vs = \expof{-t(\mI-\mP)}\vs.
\end{equation}
While it was known that estimating $\vh$ gave rise to a similar type of local Cheeger inequality~\cite{chung2009local}; until Chung and Simpson's Monte Carlo approach~\cite{chung2013solving}, no methods were known to estimate this quantity efficiently. Our new algorithm is a deterministic approach that is more suitable for comparing the properties of the diffusions. It terminates after exploring $\frac{2Ne^t}{\eps}$ edges (Theorem~\ref{thm:work}), where $N$ is a parameter that grows slowly\footnote{Since the original publication of this work we have proven an upperbound on $N$ in terms of $\eps$ and $t$~\cite{kloster2016graph}.} with $\eps$.

\textbf{Heat kernels compared to PageRank.}
These different sets of coefficients simply assign different levels of importance to walks of varying lengths: the heat kernel coefficients $\tfrac{t^k}{k!}$ decay much more quickly than $\alpha^k$, and so the heat kernel more heavily weights shorter walks in the overall sum (depicted in Figure \ref{fig:weights}). This property, in turn, will have important consequences when we study these methods in large graphs in Section~\ref{sec:experiments}.

\begin{figure}[ht]
\centering
\includegraphics[width=3in]{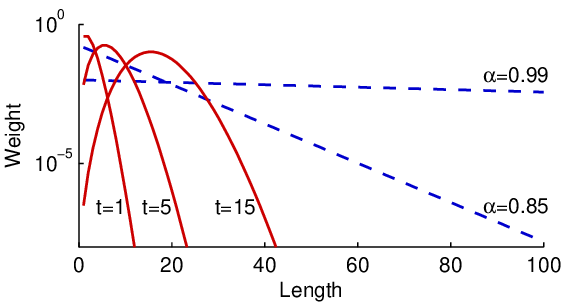}
\caption{Each curve represents the coefficients of $(\mA\mD\iv)^k$ in a sum of walks. The dotted blue lines give $\alpha^k$, and the red give $t^k/k!$, for the indicated values of $\alpha$ and $t$.}\label{fig:weights}
\end{figure}

% ALGORITHM
\section{Algorithm}\label{sec:alg}
The overall idea of the local clustering algorithm is to approximate a heat kernel vector of the form 
\[
\vh = \expof{-t(\mI - \mP) } \vs
\]
so that we can perform a sweep over $\vh$. Here we describe a coordinate-relaxation method, which we call \hkre, for approximating $\vh$. This algorithm is rooted in our recent work on computing an accurate column of $\expof{\mP}$~\cite{Kloster-2013-nexpokit}; but is heavily tuned to the objective below. Thus, while the overall strategy is classical -- just as the PageRank push method is a classic relaxation method -- the simplifications and efficient implementation are entirely novel. In particular, the new objective in this paper enables us to get a constant runtime bound independent of any property of the graph, which differs markedly from both of our previous methods~\cite{Kloster-2013-nexpokit,Kloster-preprint-nexpokit}.

\textbf{Our objective.} Recall that the final step of finding a small conductance community involves dividing by the degree of each node. Thus, our goal is to compute $\vx \approx \vh$ satisfying the degree weighted bound:
\[
\| \mD\iv\expof{-t(\mI-\mP)}\vs - \mD\iv\vx \|_{\infty} < \eps.
\]
By using standard properties of the matrix exponential, we can factor $\expof{-t(\mI-\mP)} = e^{-t}\expof{t\mP)}$ and scale by $e^t$ so that the above problem is equivalent to computing $\vy$ satisfying $\|\mD\iv(\expof{t\mP}\vs-\vy)\|_{\infty} < e^t \eps$. The element-wise characterization is that $\vy$ must satisfy:
\begin{equation}\label{terminate}
|e^t\vh_i - \vy_i| < e^t \eps d_i
\end{equation}
 for all $i$. A similar weighted objective was used in the push algorithm for PageRank~\cite{andersen2006-local}. 
 
\textbf{Outline of algorithm.} 
To accomplish this, we first approximate $\expof{t\mP}$ with its degree $N$ Taylor polynomial, $T_N(t\mP)$, and then we compute $T_N(t\mP)\vs$. But we use a large, implicitly formed linear system to avoid explicitly evaluating the Taylor polynomial. Once we have the linear system, we state a relaxation method in the spirit of Gauss-Seidel and the PageRank push algorithm in order to compute an accurate approximation of $\vh$.

% TAYLOR POLY
\subsection{Taylor Polynomial for exp\{X\}}
Determining the exponential of a matrix is a sensitive computation with a rich history~\cite{Moler-1978-19,Moler-2003-19}. For a general matrix $\mG$, an approximation via the Taylor polynomial,
\[
\expof{\mG} = \sum_{k=0}^{\infty} \tfrac{1}{k!}\mG^k \approx \sum_{k=0}^N \tfrac{1}{k!}\mG^k,
\]
can be inaccurate when $\| \mG \|$ is large and $\mG$ has mixed signs, as large powers $\mG^k$ can contain large, oppositely signed numbers that cancel properly only in exact arithmetic. However, we intend to compute $\expof{t\mP}\vs$, where $\mP$, $t$, and $\vs$ are nonnegative, so the calculation does not rely on any delicate cancellations. Furthermore, our approximation need not be highly precise. We therefore use the polynomial $\expof{t\mP}\vs \approx  T_N(t\mP) = \sum_{k=0}^N \tfrac{t^k}{k!} \mP^k$ for our approximation. For details on choosing $N$, see Section \ref{Ndetails}. For now, assume that we have chosen $N$ such that
\begin{equation} \label{eq:Nrequire}
\| \mD\iv \expof{t\mP}\vs - \mD\iv T_N(t\mP)\vs \|_{\infty} < \eps/2.
\end{equation}
This way, if we compute $\vy \approx T_N(t\mP)\vs$ satisfying
\[
\| \mD\iv T_N(t\mP)\vs - \mD\iv \vy \|_{\infty}
< \eps/2,
\]
then by the triangle inequality we will have
\begin{equation}\label{tayloraccuracy}
\| \mD\iv \expof{t\mP}\vs - \mD\iv\vy \|_{\infty} < \eps,
\end{equation}
 our objective.
 
% ERROR WEIGHTS
\subsection{Error weights}
Using a degree $N$ Taylor polynomial, \hkre ultimately approximates $\vh$ by approximating each term in the sum of the polynomial times the vector $\vs$: 
\[
\vs + \tfrac{t}{1}\mP\vs + \cdots + \tfrac{t^N}{N!}\mP^N\vs.
\]
The total error of our computed solution is then a weighted sum of the errors at each individual term, $\tfrac{t^k}{k!}\mP^k\vs$.  We show in Lemma \ref{lem:psiweights} that these weights are given by the polynomials $\psi_k(t)$, which we define now. For a fixed degree $N$ Taylor polynomial of the exponential, $T_N = \sum_{k=0}^N \tfrac{t^k}{k!}$, we define
\begin{equation}\label{defpsi}
\psi_k := \sum_{m=0}^{N-k} \tfrac{k!}{(m+k)!}t^{m} \text{ for $k = 0, \cdots, N$}.
\end{equation}
These polynomials $\psi_k(t)$ are closely related to the $\phi$ functions central to exponential integrators in ODEs~ \cite{minchev2005norges}. Note that $\psi_0 = T_N$.

To guarantee the total error satisfies the criterion \eqref{terminate} then, it is enough to show that the error at each Taylor term satisfies an $\infty$-norm inequality analogous to \eqref{terminate}. This is discussed in more detail in Section \ref{sec:theory}.

% LINEAR SYSTEM 
\subsection{Deriving a linear system}\label{sec:linsys}
To define the basic step of the \hkre  algorithm and to show how the $\psi_k$ influence the total error, we rearrange the Taylor polynomial computation into a linear system.

Denote by $\vv_k$ the $k^{\text{th}}$ term of the vector sum $T_N(t\mP)\vs$:
\begin{align}
T_N(t\mP)\vs &= \vs + \tfrac{t}{1}\mP\vs + \cdots + \tfrac{t^N}{N!}\mP^N\vs\\
 &= \vv_0 + \vv_1 + \cdots + \vv_N.
\end{align}
Note that $\vv_{k+1} = \tfrac{t^{k+1}}{(k+1)!}\mP^{k+1} = \tfrac{t}{(k+1)}\mP\vv_k$. This identity implies that the terms $\vv_k$ exactly satisfy the linear system
\begin{equation}\label{largelinsys}
\left[\begin{array}{ccccc}
  \mI & & & & \\
 \tfrac{-t}{1}\mP & \mI & & & \\
        & \tfrac{-t}{2}\mP & \ddots & & \\
      & & \ddots& \mI&  \\
      &    &    & \tfrac{-t}{N}\mP & \mI \\
 \end{array}\right]
\bmat{ \vv_0 \\ \vv_1 \\ \vdots \\ \vdots \\ \vv_N } = \bmat{\vs \\ 0 \\ \vdots \\ \vdots \\ 0} .
\end{equation}
Let $\vv = [ \vv_0; \vv_1; \cdots; \vv_N]$. An approximate solution $\hat{\vv}$ to \eqref{largelinsys} would have block components $\hat{\vv}_k$ such that $\sum_{k=0}^N \hat{\vv}_k \approx T_N(t\mP)\vs$, our desired approximation to $e^t \vh$. In practice, we update only a single length $n$ solution vector, adding all updates to that vector, instead of maintaining the $N+1$ different block vectors $\hat{\vv}_k$ as part of $\hat{\vv}$; furthermore, the block matrix and right-hand side are never formed explicitly.

With this block system in place, we can describe the algorithm's steps.

\subsection{The hk-relax algorithm}\label{sec:relax}
Given  a random walk transition matrix $\mP$, scalar $t>0$, and seed vector $\vs$ as inputs, we solve the linear system from \eqref{largelinsys} as follows. Denote the initial solution vector by $\vy$ and the initial $nN \times 1$ residual by $\vr^{(0)} = \ve_1\kron\vs$. Denote by $r(i,j) $ the entry of $\vr$ corresponding to node $i$ in residual block $j$. The idea is to iteratively remove all entries from $\vr$ that satisfy
\begin{equation}\label{resqueue}
r(i,j) \geq \frac{e^t \eps d_i}{ 2N \psi_j(t)}.
\end{equation}
To organize this process, we begin by placing the nonzero entries of $\vr^{(0)}$ in a queue, $Q(\vr)$, and place updated entries of $\vr$ into $Q(\vr)$ only if they satisfy \eqref{resqueue}.

Then \hkre  proceeds as follows.
\begin{compactenum}
\item At each step, pop the top entry of $Q(\vr)$, call it $r(i,j)$, and subtract that entry in  $\vr$, making $\vr(i,j)=0$.
\item Add $r(i,j)$ to $\vy_i$.
\item Add $r(i,j)\tfrac{t}{j+1}\mP\ve_i$ to residual block $\vr_{j+1}$.
\item For each entry of $\vr_{j+1}$ that was updated, add that entry to the back of $Q(\vr)$ if it satisfies \eqref{resqueue}.
\end{compactenum}

Once all entries of $\vr$ that satisfy \eqref{resqueue} have been removed, the resulting solution vector $\vy$ will satisfy \eqref{terminate}, which we prove in Section \ref{sec:theory}, along with a bound on the work required to achieve this. We present working code for our method in Figure~\ref{fig:pseudocode} that shows how to optimize this computation using sparse data structures. These make it highly efficient in practice.

\begin{figure}[ht]
\begin{lstlisting}[language=python,commentstyle={\color{blue}\itshape}]
# G is graph as dictionary-of-sets, 
# seed is an array of seeds,
# t, eps, N, psis are precomputed
x = {} # Store x, r as dictionaries
r = {} # initialize residual
Q = collections.deque() # initialize queue
for s in seed: 
  r[(s,0)] = 1./len(seed)
  Q.append((s,0))
while len(Q) > 0:
  (v,j) = Q.popleft() # v has r[(v,j)] ...
  rvj = r[(v,j)]
  # perform the hk-relax step
  if v not in x: x[v] = 0.
  x[v] += rvj 
  r[(v,j)] = 0.
  mass = (t*rvj/(float(j)+1.))/len(G[v])
  for u in G[v]:   # for neighbors of v
    next = (u,j+1) # in the next block
    if j+1 == N: # last step, add to soln
      x[u] += rvj/len(G(v))
      continue
    if next not in r: r[next] = 0.
    thresh = math.exp(t)*eps*len(G[u])
    thresh = thresh/(N*psis[j+1])/2.
    if r[next] < thresh and \
       r[next] + mass >= thresh:
       Q.append(next) # add u to queue
    r[next] = r[next] + mass
\end{lstlisting}
\caption{Pseudo-code for our algorithm as 
working python code. The graph is stored
as a dictionary of sets so that the 
{\normalfont\texttt{G[v]}} statement returns the set of
neighbors associated with vertex $v$. The 
solution is the vector $x$ indexed by 
vertices and the residual vector is 
indexed by tuples $(v,j)$ that are
pairs of vertices and steps $j$.
A fully working demo may be downloaded from github
\small\normalfont \texttt{https://gist.github.com/dgleich/7d904a10dfd9ddfaf49a}.
}
\label{fig:pseudocode}
\end{figure}

% CHOOSING N
\subsection{Choosing N}\label{Ndetails}
The last detail of the algorithm we need to discuss is how to pick $N$. In \eqref{eq:Nrequire} 
we want to guarantee the accuracy of
$\mD\iv\expof{t\mP}\vs - \mD\iv T_N(t\mP)\vs $. By using
$\mD\iv\expof{t\mP} = \expof{t\mP^T}\mD\iv$ and $\mD\iv T_N(t\mP) = T_N(t\mP^T)\mD\iv$,
we can get a new upperbound on $\| \mD\iv\expof{t\mP}\vs - \mD\iv T_N(t\mP)\vs \|_{\infty}$ by noting
\[ \begin{aligned}
& \| \expof{t\mP^T}\mD\iv\vs - T_N(t\mP^T)\mD\iv\vs \|_\infty  \\
& \qquad \le \|\expof{t\mP^T} - T_N(t\mP^T)\|_{\infty} \|\mD\iv \vs\|_{\infty}.
\end{aligned} \]
Since $\vs$ is stochastic, we have $\| \mD\iv \vs\|_{\infty} \leq \|\mD\iv \vs \|_1 \leq 1$. From~\cite{liou1966novel} we know that the norm $\|\expof{t\mP^T} - T_N(t\mP^T)\|_{\infty}$ is bounded by
\begin{equation}
\frac{\|t\mP^T\|_{\infty}^{N+1}}{(N+1)!}\frac{(N+2)}{(N+2-t)} \leq \frac{t^{N+1}}{(N+1)!}\frac{(N+2)}{(N+2-t)}.
\end{equation} 
So to guarantee \eqref{eq:Nrequire}, it is enough to choose $N$ that implies $\frac{t^{N+1}}{(N+1)!}\frac{(N+2)}{(N+2-t)} < \eps/2$. Such an $N$ can be determined efficiently simply by iteratively computing terms of the Taylor polynomial for $e^t$ until the error is less than the desired error for \hkre. In practice, this required a choice of $N$ no greater than $2t\log(\tfrac{1}{\eps})$, which we think can be made rigorous.

\subsection{Outline of convergence result}

The proof proceeds as follows. First, we relate the error vector of the Taylor approximation $E_1 = \tnps - \vx$, to the error vector from solving the linear system described in Section \ref{sec:relax}, $E_2 = \tnps - \vy$. Second, we express the error vector $E_2$ in terms of the residual blocks of the linear system \eqref{largelinsys}; this will involve writing $E_2$ as a sum of residual blocks $\vr_k$ with weights $\psi_k(t\mP)$. Third, we use the previous results to upperbound
$\|\mD\iv\tnps - \mD\iv\vx\|_{\infty}$ with $\sum_{k=0}^N \psi_k(t) \| \mD\iv \vr_k\|_{\infty}$, and use this to show that $\|\mD\iv\tnps - \mD\iv\vx\|_{\infty} < \eps/2$ is guaranteed
by the stopping criterion of \hkre, \eqref{terminate}. Finally, we prove that performing steps of \hkre until the stopping criterion is attained requires work bounded by $\frac{2N \psi_1(t)}{\eps} \le 2N e^{t} / \eps$.

\section{Related Work}
\label{sec:related}
We wish to highlight a few ideas that have recently emerged in the literature to clarify how our method differs. We discuss these in terms of community detection, the matrix exponential, fast diffusion methods, and relaxation methods.

\textbf{Community detection and conductance.}
Conductance often appears in community detection and is known to be one of the most important measures of a community~\cite{Schaeffer-2007-clustering}. The personalized PageRank method is one of the most scalable methods to find sets of small conductance, although recent work opened up a new possibility with localized max-flow algorithms~\cite{Orecchia-2014-flow}. For the PageRank algorithm we use as a point of comparison, Zhu et al.~\cite{Zhu-2013-local} recently provided an improved bound on the performance of this algorithm when finding sets with high internal conductance. The internal conductance of a set is the minimum conductance of the subgraph induced by the set and we would expect that real-world communities have large internal conductance. Due to the similarity between our algorithm and the personalized PageRank diffusion, we believe that a similar result likely holds here as well. 

\textbf{The matrix exponential in network analysis.}
Recently, the matrix exponential has frequently appeared as a tool in the network analysis literature. 
It has been used to estimate node centrality~\cite{Estrada-2000-index,farahat2006-exphits,Estrada-2010-matrix-functions}, for link-prediction~\cite{Kunegis-2009-learning-spectral}, in graph kernels~\cite{Kondor-2002-diffusion}, and -- as already mentioned -- clustering and community detection~\cite{chung2007-pagerank-heat}. Many of these studies involve fast ways to approximate the \emph{entire} matrix exponential, instead of a single column as we study here. For instance, Sui et al.~\cite{Sui-2013-low-rank} describe a low-parameter decomposition of a network that is useful both for estimating Katz scores~\cite{katz1953-status} and the matrix exponential. Orecchia and Mahoney~\cite{Orecchia-2011-implicit-regularization} show that the heat kernel diffusion implicitly approximates a diffusion operator using a particular type of generalized entropy, which provides a principled rationale for its use.

\textbf{Fast methods for diffusions.}
Perhaps the most related work is a recent Monte Carlo algorithm by Chung and Simpson~\cite{chung2013solving} to estimate the heat kernel diffusion via a random walk sampling scheme. This approach involves directly simulating a random walk with transition probabilities that mirror the Taylor series expansion of the exponential. In comparison, our approach is entirely deterministic and thus more useful to compare between diffusions, as it eliminates the algorithmic variance endemic to Monte Carlo simulations. A similar idea is used by Borgs et al.~\cite{Borgs-2013-sublinear} to achieve a randomized sublinear time algorithm to estimate the largest PageRank entries, and in fact, Monto Carlo methods frequently feature in PageRank computations due to the relationship between the diffusion and a random walk~\cite{Avrachenkov-2007-Monte-Carlo,Borgs-2013-sublinear,Bahmani-2011-PageRank-MR,Bahmani-2010-PageRank}. Most other deterministic approaches for the matrix exponential involve at least one matrix-vector product~\cite{Orecchia-2012-exponential,Al-Mohy-2011-exponential}.

\textbf{Relaxation methods.} The algorithm we use is a coordinate relaxation method similar to Gauss-Seidel and Gauss-Southwell. If we applied it to a symmetric positive definite matrix, then it would be a coordinate descent method~\cite{Luo1992-coordinate-descent}. It has been proposed for PageRank in a few difference cases~\cite{jeh2003-personalized,mcsherry2005-uniform,andersen2006-local}. The same type of relaxation method has also been used to estimate the Katz diffusion~\cite{Bonchi-2012-fast-katz}. We recently used it to estimate a column of the matrix exponential $\expof{\mP} \ve_i$ in a strict, 1-norm error and were able to prove a sublinear convergence bound by assuming a very slowly growing maximum degree~\cite{Kloster-2013-nexpokit} or a power-law degree distribution~\cite{Kloster-preprint-nexpokit}. This paper, in comparision, treats the scaled exponential $\expof{-t(\eye - \mP)} \ve_i$ in a degree-weighted norm; it also shows a constant runtime independent of any network property.

% EXPERIMENTS
\section{Experimental Results}
\label{sec:experiments}

Here we compare \hkre with a PageRank-based local clustering algorithm, pprpush~\cite{andersen2006-local}. Both algorithms accept as inputs a symmetric graph $\mA$ and seed set $\vs$. The parameters required are $t$ and $\eps$, for \hkre, and $\alpha$ and $\eps$ for pprpush. Both algorithms compute their respective diffusion ranks starting from the seed set, then perform a sweep-cut on the resulting ranks. The difference in the algorithms lies solely in the diffusion used and the particular parameters. We conducted the timing experiments on a Dual CPU system with the Intel Xeon E5-2670 processor (2.6 GHz, 8 cores) with 16 cores total and 256 GB of RAM. None of the experiments needed anywhere near all the memory, nor any of the parallelism. Our implementation uses Matlab's sparse matrix data structure through a C++ mex interface. It uses C++ unordered maps to store sparse vectors and is equivalent to the code in Figure~\ref{fig:pseudocode}.

\subsection{Synthetic results}
In this section, we study the behavior of the PageRank and heat kernel diffusions on the symbolic image graph of a chaotic function $f$~\cite{Shepelyansky2010-Ulam}. The graphs that result are loosely reminiscent of a social network because they have pockets of structure, like communities, and also chaotic behaviour that results in a small-world like property.

The symbolic image of a function $f$ is a graph where each node represents a region of space, and edges represent the action of the function in that region of space. We consider two-dimensional functions $f(x,y)$ in the unit square $[0,1]^2$ so that we can associate each node with a pixel of an image and illustrate the vectors as images. In Figure~\ref{fig:chirikov} (left), we illustrate how the graph construction works. In the remaining examples, we let $f$ be a map that results from a chaotic, nonlinear dynamical system~\cite{Shepelyansky2010-Ulam} (we use the T10 construction with $k=0.22,\eta=0.99$ and sample 1000 points from each region of space, then symmetrize the result and discard weights). We discretize space in a $512 \times 512$ grid, which results in a $262,144$ node graph with $2M$ edges. In Figure~\ref{fig:chirikov} (right), we also show the PageRank vector with uniform teleportation as an image to ``illustrate the structure'' in the function $f$.

Next, in Figure~\ref{fig:chirikov-ex}, we compare the vectors and sets identified by both diffusions starting from a single seed node. We chose the parameters $\alpha = 0.85$ and $t=3$ so the two methods perform the same amount of work. These results are what would be expected from Figure~\ref{fig:weights}, and what many of the remaining experiments show: PageRank diffuses to a larger region of the graph whereas the heat-kernel remains more focused in a sub-region. PageRank, then, finds a large community with about 5,000 nodes whereas the heat kernel finds a small community with around $452$ nodes with slightly worse conductance. This experiment suggests that, if these results also hold in real-world networks, then because real-world communities are often small~\cite{Leskovec-2009-community-structure}, the heat kernel diffusion should produce \emph{more accurate} communities in real-world networks.

\begin{figure}[ht]
\centering
\hspace{1em}
\includegraphics[width=0.4\linewidth]{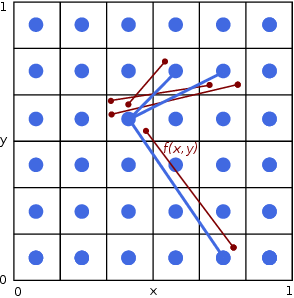}\hfill  \includegraphics[width=0.4\linewidth]{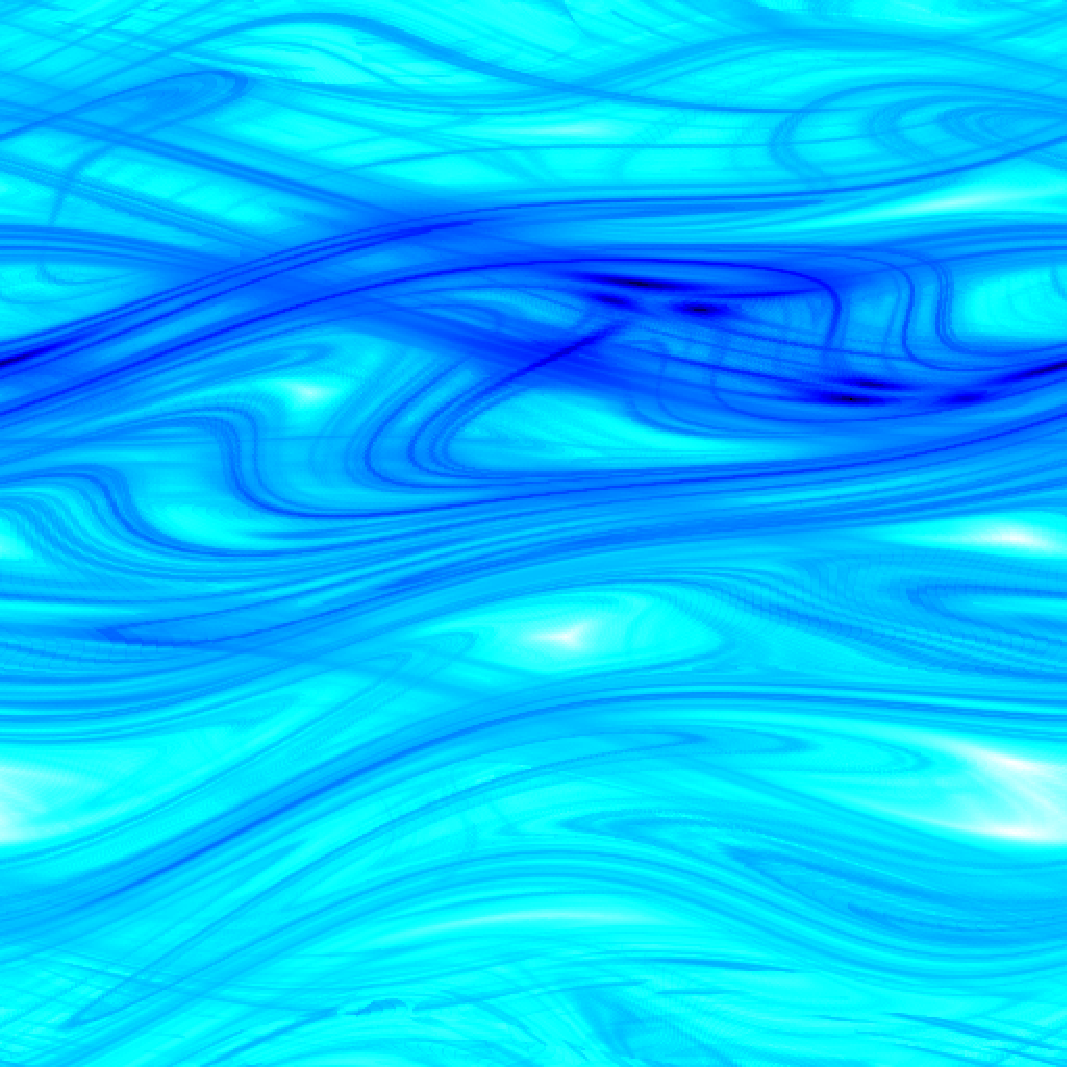}
\hspace*{1em}
\caption{(Left) An illustration of the symbolic image of a function $f$ as a graph. Each
large, blue node represents a region of space. Thick, blue edges represent how the
thin, red values of the function behave in that region. (Right) The global PageRank vector is then
an image that illustrates features of the chaotic map $f$ and shows that it has pockets
of structure.}
\label{fig:chirikov}

\end{figure}

\begin{figure}[ht] \centering
\fontsize{7}{8}\selectfont%
\begin{tabular}{@{}c@{}c@{}}
\includegraphics[width=0.5\linewidth]{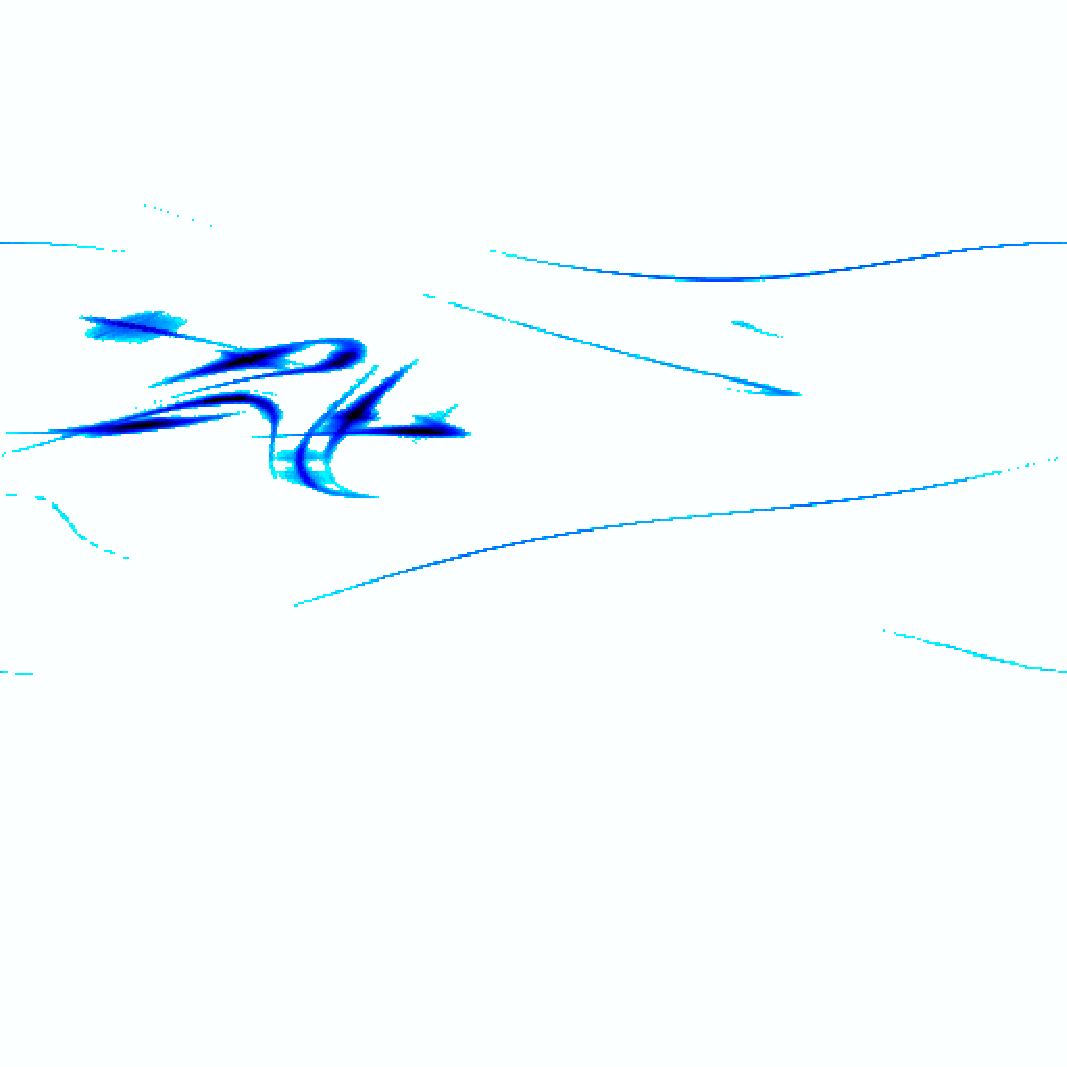} &
\includegraphics[width=0.5\linewidth]{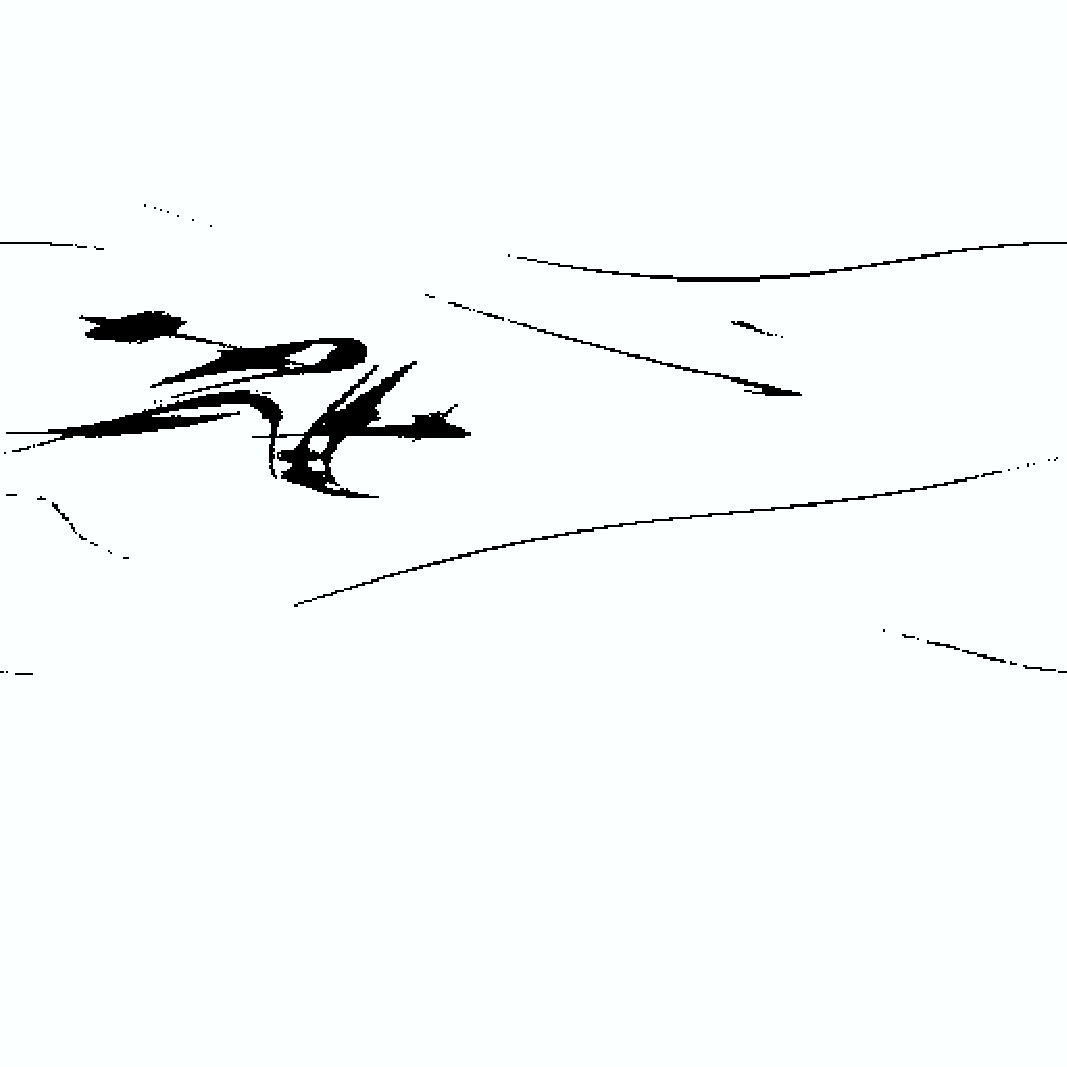}\\
 (a) ppr vector $\alpha = 0.85$, $\eps = 10^{-5}$ & (b) ppr set, $\phi = 0.31$, size = 5510 \\
\includegraphics[width=0.5\linewidth]{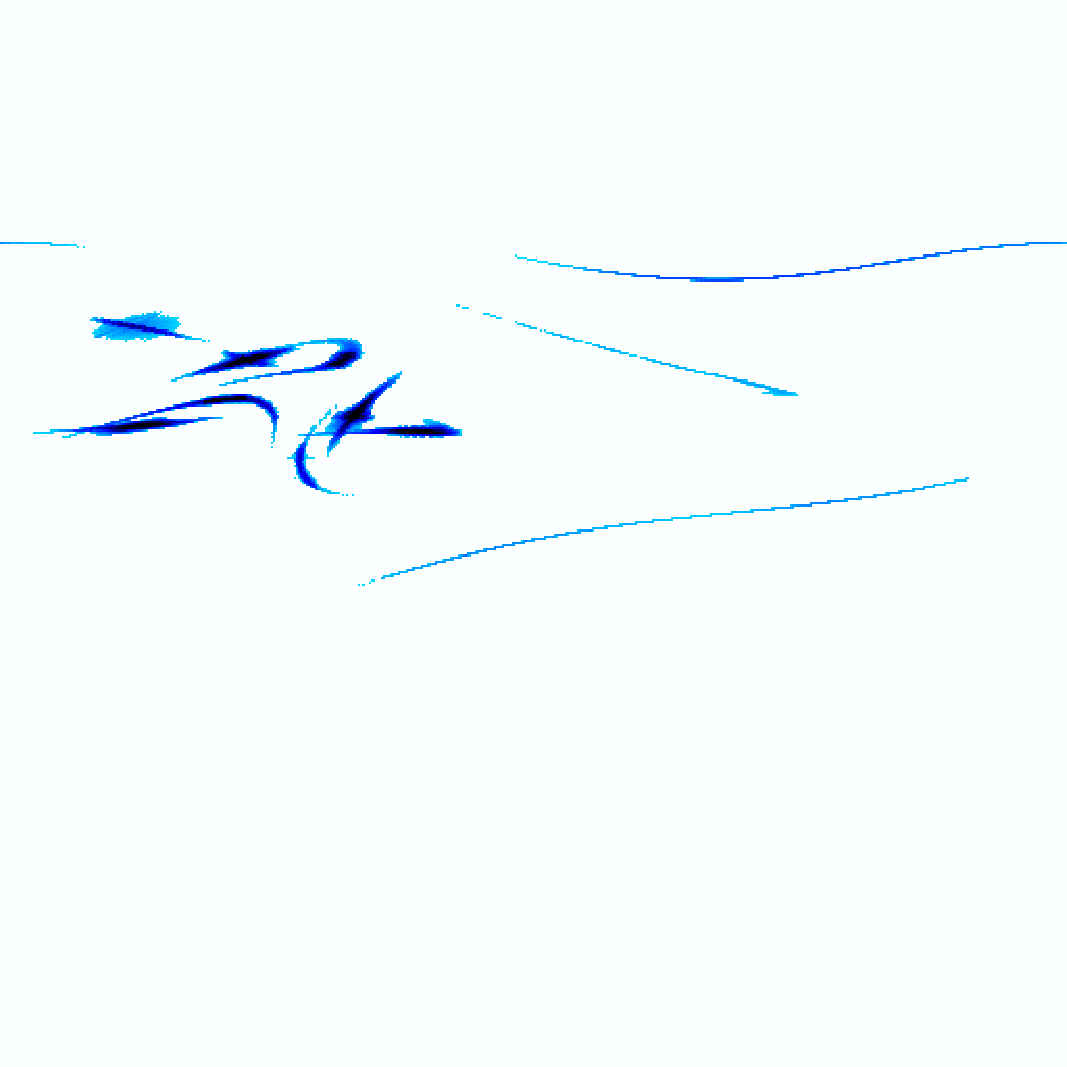} &
\includegraphics[width=0.5\linewidth]{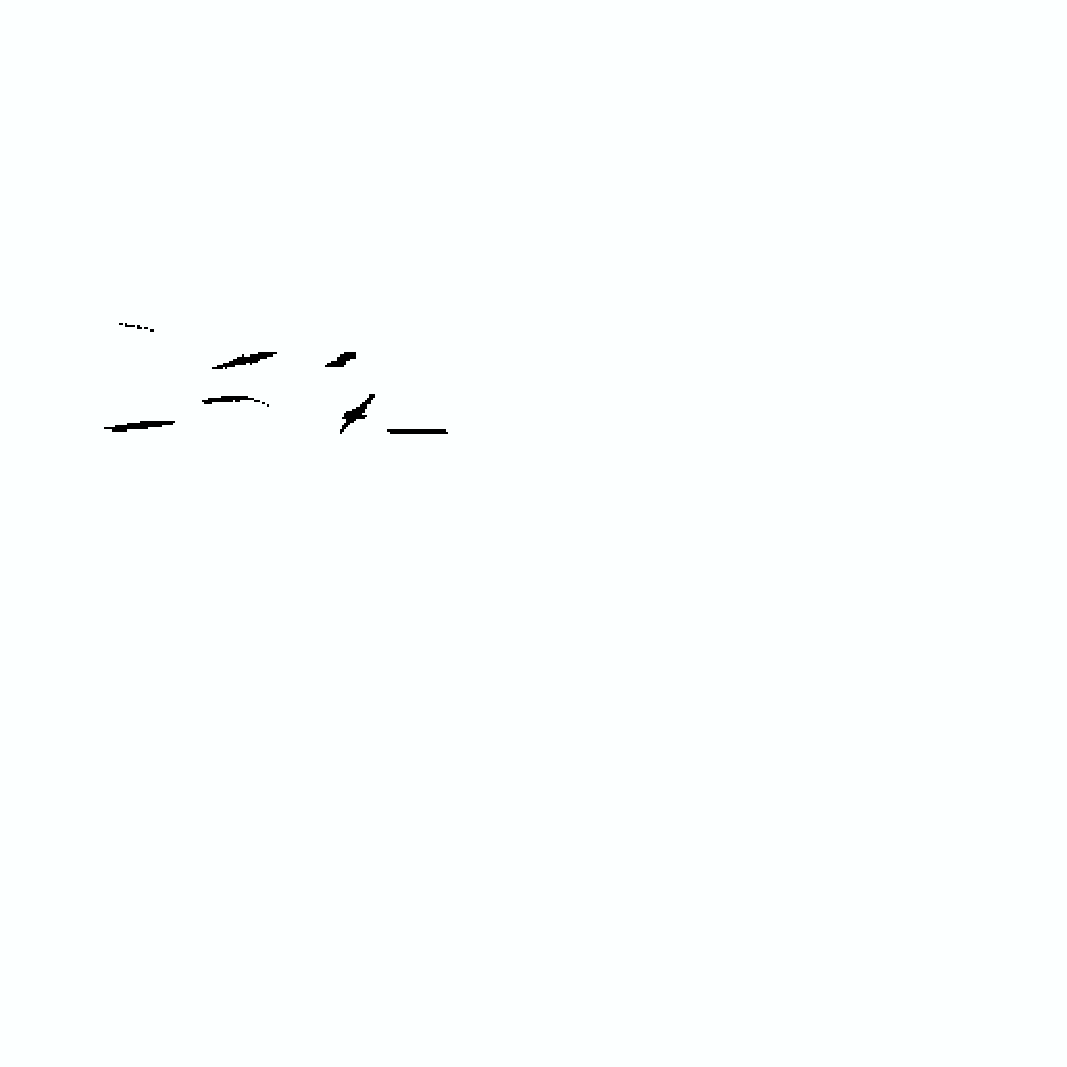}\\
 (c) hk vector $t = 3$, $\eps = 10^{-5}$ & (d) hk set, $\phi = 0.34$, size = 455\\
\end{tabular}
\caption{When we compare the heat-kernel and PageRank diffusions on the symbolic image
of the Chirikov map (see Figure~\ref{fig:chirikov}), pprgrow finds a larger set
with slightly better conductance, whereas hkgrow finds a tighter set with about the
same conductance. In real-world networks, these smaller sets are more like real-world
communities.}
\label{fig:chirikov-ex}
\end{figure}

\subsection{Runtime and conductance}

We next compare the runtime and conductance of the algorithms on a suite of social networks. For pprpush, we fix $\alpha = 0.99$, then compute PageRank for multiple values of $\eps = 10^{-2}, 10^{-3}, 10^{-4}, 10^{-5}$, and output the set of best conductance obtained. (This matches the way this method is commonly used in past work.) For \hkre, we compute the heat kernel rank for four different parameter sets, and output the set of best conductance among them: $(t,\eps) = (10,10^{-4}); (20,10^{-3}); (40,5\cdot 10^{-3}); (80,10^{-2})$. We also include in \hkre an early termination criterion, in the case that the sum of the degrees of the nodes which are relaxed, $\sum d_{i_l}$, exceeds $n^{1.5}$.
However, even the smaller input graphs (on which the condition is more likely to be met because of the smaller value of $n^{1.5}$) do not appear to have reached this threshold. Furthermore, the main theorem of this paper implies that the quantity $\sum d_{i_l}$ cannot exceed $\frac{2N \psi_1(t)}{\eps}$. The datasets we use are summarized in Table~\ref{tab:datasets}; all datasets are modified to be undirected and a single connected component. These datasets were originally presented in the following papers~\cite{alberich2002marvel,Moguna-2004-models,Boldi-2011-layered,Kwak2010-Twitter,Leskovec-2010-signed,Leskovec-2007-densification,Leskovec-2009-community-structure,mislove2007measurement,Newman-2006-network,Newman-2001-scientific-collaboration,Caida-2005-network,Yang:2012:DEN:2350190.2350193}.

\begin{table}
\centering
\caption{Datasets} \label{tab:datasets}
\begin{tabularx}{\linewidth}{lrr} \toprule
Graph& $|V|$ &$|E|$\\ \midrule
\texttt{pgp-cc} &     10,680  &  24,316 \\
\texttt{ca-AstroPh-cc}&   17,903  &    196,972 \\
\texttt{marvel-comics-cc}&        19,365   &   96,616\\
\texttt{as-22july06}&           22,963   &    48,436\\
\texttt{rand-ff-25000-0.4}&       25,000   &   56,071\\
\texttt{cond-mat-2003-cc}&         27,519   &   116,181\\
\texttt{email-Enron-cc}&         33,696   &   180,811 \\
\texttt{cond-mat-2005-fix-cc}&  36,458   &  171,735 \\
\texttt{soc-sign-epinions-cc}&  119,130   &  704,267\\
\texttt{itdk0304-cc}&       190,914   &  607,610\\
\texttt{dblp-cc}&        226,413   &  716,460\\
\texttt{flickr-bidir-cc}&        513,969   &  3,190,452 \\
\texttt{ljournal-2008} &  5,363,260 &  50,030,085 \\

\texttt{twitter-2010} &  41,652,230   &  1,202,513,046 \\
\texttt{friendster} &  65,608,366   & 1,806,067,135 \\
\hline
\texttt{com-amazon} &  334,863   & 925,872 \\
\texttt{com-dblp} & 317,080 & 1,049,866 \\
\texttt{com-youtube} &  1,134,890   & 2,987,624 \\
\texttt{com-lj} &  3,997,962 &  34,681,189  \\
\texttt{com-orkut} &  3,072,441   & 117,185,083 \\
 \bottomrule
\end{tabularx}
\end{table}

To compare the runtimes of the two algorithms, we display in Figure \ref{fig:timecond} for each graph the $25\%, 50\%,$ and $75\%$ percentiles of the runtimes from 200 trials performed. For a given graph, each trial consisted of choosing a node of that graph uniformly at random to be a seed, then calling both the PageRank and the heat kernel algorithms. On the larger datasets, which have a much broader spectrum of node degrees and therefore greater variance, we instead performed 1,000 trials. Additionally, we display the $25\%, 50\%,$ and $75\%$ percentiles of the conductances achieved during the exact same set of trials. The trendlines of these figures omit some of the trials in order to better show the trends, but all of the median results are plotted (as open circles). The figures show that in small graphs, \hkre is faster than pprpush, but gets larger (worse) conductance. The picture reverses for large graphs where \hkre is slower but finds smaller (better) conductance sets.

\begin{figure}[ht]
\centering
\includegraphics[width=3in]{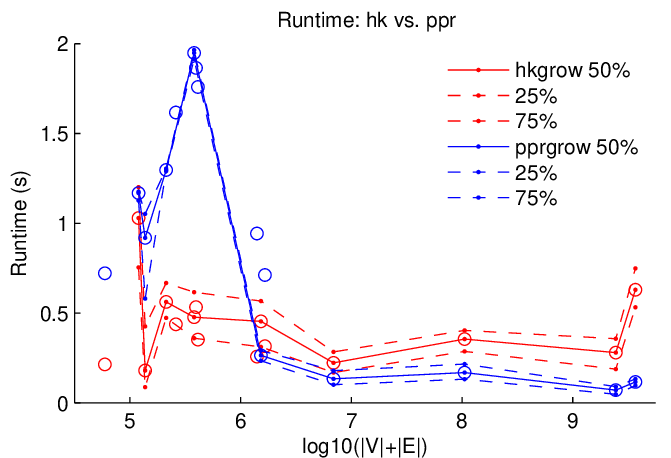}
\includegraphics[width=3in]{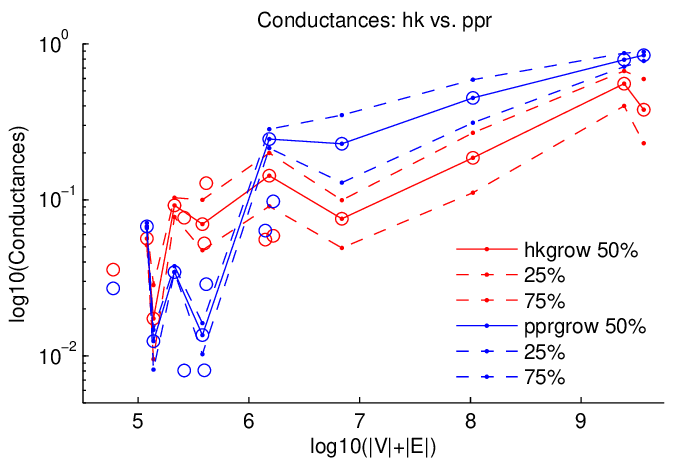}
\caption{(Top figure) Runtimes of the hk-relax vs.~ppr-push, shown with percentile trendlines from a select set of experiments. (Bottom) Conductances of hk-relax vs.~ppr-push, shown in the same way.}
\label{fig:timecond}
\end{figure}

\textbf{Cluster size and conductance.}
We highlight the individual results of the previous experiment on the symmetrized twitter network. Here, we find that \hkre finds sets of better conductance than pprpush at all sizes of communities in the network.  See Figure~\ref{fig:condsize}.
\begin{figure}[ht]
\centering
\includegraphics[width=3in]{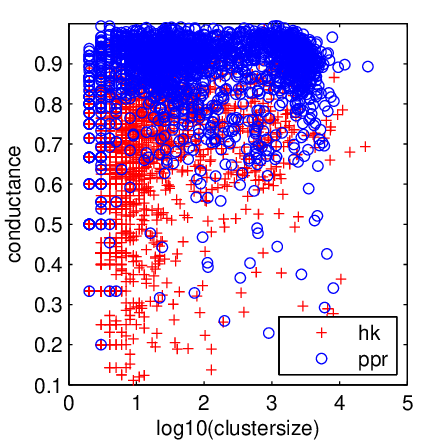}
\includegraphics[width=3in]{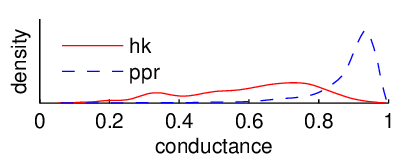}
\caption{The top figure shows a scatter plot of conductance vs.~community size in the twitter graph for the two community detection methods; the bottom figure shows a kernel density estimate of the conductances achieved by each method, which shows that hk-relax is more likely to return a set of lower conductance.}
\label{fig:condsize}
\end{figure}

\subsection{Clusters produced vs.~ground-truth}

\begin{table}[t]
\caption{The result of evaluating the heat kernel (hk) vs.~PageRank (pr) on finding real-world communities. The heat kernel finds smaller, more accurate, sets with slightly worse conductance.}
\label{tab:ground-truth}
 \begin{tabularx}{1\linewidth}{lXXXXXX} \toprule
 %\multicolumn{1}{c}{} & & & & & & \\[\dimexpr-\normalbaselineskip-\arrayrulewidth]% Correct for mis-alignment
data & \multicolumn{2}{l}{$F_1$-measure} & \multicolumn{2}{l}{conductance} & \multicolumn{2}{l}{set size} \\
 & \texttt{hk} & \texttt{pr} &
 \texttt{hk} & \texttt{pr} &
 \texttt{hk} & \texttt{pr} \\
\midrule
% CORRECTED:
\texttt{amazon} & 0.608 & 0.415 & 0.124  &0.050 & 145& 5,073\\
\texttt{dblp} & 0.364& 0.273& 0.238 & 0.144 & 156& 4,529\\
\texttt{youtube} & 0.128 & 0.078 & 0.477 & 0.361 & 137 & 5,833 \\
\texttt{lj} & 0.138 & 0.104 & 0.564 & 0.524 & 156 & 299 \\
\texttt{orkut} & 0.183 & 0.116 & 0.824 & 0.736 & 95 & 476 \\
\texttt{friendster} & 0.125 & 0.112 & 0.866 & 0.860 & 96 & 218 \\
% THESE RESULTS ARE TRUE, BUT COMPUTED ON GROUND TRUTH
% COMMUNITIES THAT INCLUDED SIZES <= 10, WHICH OUR
% EXPERIMENTAL SETUP WAS SUPPOSED TO PRECLUDE
% \texttt{amazon} & 0.325 & 0.140 & 0.141  &0.048 & 193& 15293\\
% \texttt{dblp} & 0.257& 0.115& 0.267&0.173 & 44& 16026\\
% \texttt{youtube} & 0.177& 0.136& 0.337& 0.321& 1010&6079 \\
% \texttt{lj} &0.131 & 0.107 & 0.474 & 0.459 & 283 & 738 \\
% \texttt{orkut} & 0.055 & 0.044& 0.714 & 0.687 & 537 & 1989\\
% \texttt{friendster} &0.078 & 0.090 & 0.785 & 0.802 & 229 & 333 \\
%
% com-amazon \hfill \texttt{hk}& 0.325 & 0.141 & 193 \\
% \hfill \texttt{ppr}& 0.140 & 0.048 & 15293 \\
%com-dblp \hfill \texttt{hk}& 0.257 & 0.267 & 44 \\
% \hfill \texttt{ppr}& 0.115 & 0.173 & 16026 \\
%com-youtube \hfill \texttt{hk}& 0.177 & 0.337 & 1010 \\
% \hfill \texttt{ppr}& 0.136 & 0.321 & 6079 \\
%com-lj \hfill \texttt{hk}& 0.131 & 0.474 & 283 \\
% \hfill \texttt{ppr}& 0.107 & 0.459 & 738 \\
%com-orkut \hfill \texttt{hk}& 0.055 & 0.714 & 537 \\
% \hfill \texttt{ppr}& 0.044 & 0.687 & 1989 \\
%com-friendster \hfill \texttt{hk}& 0.078 & 0.785 & 229 \\
% \hfill \texttt{ppr}& 0.090 & 0.802 & 333 \\
%
%com-amazon & 0.140124 & 0.325270 & \\
%com-dblp & 0.115376 & 0.256637 & \\
%com-youtube &  0.135619 & 0.176799 & \\
%com-lj & 0.107059 &  0.131054 & \\
%com-orkut & 0.043852 & 0.054912 & \\
%com-friendster & 0.090118 & 0.078045 & \\
\bottomrule
\end{tabularx}
\end{table}
We conclude with an evaluation of identifying ground-truth communities in the \texttt{com-dblp}, \texttt{com-lj}, \texttt{com-amazon}, \texttt{com-orkut},  \texttt{com-youtube}, and \texttt{com-friendster} datasets~\cite{mislove2007measurement,Yang:2012:DEN:2350190.2350193}. In this experiment, for each dataset we first located 100 known communities in the dataset of size greater than 10\footnote{We thank AmirMahdi Ahmadinejad for pointing out a small mistake in our original implementation of this experiment -- though we reported in Table~\ref{tab:ground-truth} in our original publication~\cite{kloster2014heat} that we carried out this experiment on 100 commmunities of size greater than 10, Ahmadinejad noted that in practice we had used some communities of size less than 10. The results shown here in Table~\ref{tab:ground-truth} are from the corrected experiment, and in fact are even more favorable for our \hkre algorithm than our original reported results.}
. Given one such community, using every single node as an individual seed, we looked at the sets returned by \hkre with $t=5, \eps=10^{-4}$ and pprpush using the standard procedure. We picked the set from the seed that had the highest $F_1$ measure. (Recall that the $F_1$ measure is a harmonic mean of precision and recall.)  We report the mean of the $F_1$ measure, conductance, and set size, where the average is taken over all 100 trials in Table~\ref{tab:ground-truth}.
These results show that \hkre produces only slightly inferior conductance scores, but using much smaller sets with substantially better $F_1$ measures. This suggests that \hkre better captures the properties of real-world communities than the PageRank diffusion in the sense that the tighter sets produced by the heat kernel are better focused around real-world communities than are the larger sets produced by the PageRank diffusion.

%To compare the performance of the algorithms rigorously, we use the F-measure, which is  as follows. Given an approximation, $A$, of a set, $S$, the \textit{recall} of $A$ is $|A \cap S| / |A|$, or how many of your "guesses" were correct. The \textit{precision} of $A$ is $|A\cap S|/|S|$, or how many of the total "answers" were able to guess. The F-measure of the set $A$ as an approximation of $S$ is then
%\[
%F(A) = \frac{2 \cdot \text{precision}\cdot \text{recall}}{\text{precision} + \text{recall}}.
%\]
%We use the F-measure because it rewards an algorithm for capturing a large part of the community, but not if the algorithm accomplishes this by returning \textit{too} large a cluster. This experiment shows that, given a seed node inside a ground-truth community, the clusters determined by the heat kernel diffusion using that seed node would correlate slightly better with the ground-truth community than those returned by PageRank diffusion.
%\begin{table}
%\centering
%\caption{Listed are the mean f-measures of the clusters found by each algorithm, as compared with the ground-truth communities the seed nodes were taken from.{\normalfont\texttt{ljournal-2008}} dataset.}

%\end{table}

% CONCLUSIONS
\section{Conclusions}

These results suggest that the \hkre algorithm is a viable companion to the celebrated PageRank push algorithm and may even be a worthy competitor for tasks that require accurate communities of large graphs. Furthermore, we suspect that the \hkre method will be useful for the myriad other uses of PageRank-style diffusions such as link-prediction~\cite{Kunegis-2009-learning-spectral} or even logic programming~\cite{Wang:2013:PPP:2505515.2505573}.

In the future, we plan to explore this method on directed networks as well as better methods for selecting the parameters of the diffusion. It is also possible that our new ideas may translate to faster methods for non-conservative diffusions such as the Katz~\cite{katz1953-status} and modularity methods~\cite{Newman-2006-modularity}, and we plan to explore these diffusions as well.


\section*{Acknowledgements}
This work was supported by NSF CAREER Award CCF-1149756.

\fontsize{8}{9}\selectfont
\bibliographystyle{abbrv}
\bibliography{all-bibliography,newbib}  % sigproc.bib is the name of the Bibliography in this case

\appendix
\section{Convergence theory}\label{sec:theory}\label{sec:anls}

Here we state our main result bounding the work required by \hkre  to approximate the heat kernel with accuracy as described in \eqref{terminate}.
\begin{theorem}\label{thm:work}
Let $\mP$, $t$, $\psi_k(t)$, and $\vr$ be as in Section \ref{sec:alg}. If steps of \hkre  are performed until all entries of the residual satisfy $r(i,j) < \frac{e^t \eps d_i}{ 2 N \psi_j(t)}$, then \hkre  produces an approximation $\vx$ of $\vh = \expof{-t(\mI-\mP)}\vs$ satisfying 
\[
\| \mD\iv\expof{-t(\mI-\mP)}\vs -\mD\iv\vx \|_{\infty} < \eps,
\]
and the amount of work required is bounded by
\[
\text{work}(\eps) \leq \frac{2N\psi_1(t)}{\eps} \leq \frac{2N (e^t-1)}{\eps t}.
\]
\end{theorem}

Producing $\vx$ satisfying 
\[
\|\mD\iv\expof{-t(\mI - \mP)}\vs - \mD\iv\vx \|_{\infty} < \eps
\]
 is equivalent to producing $\vy$ satisfying 
 \[
 \|\mD\iv\expof{t\mP)}\vs - \mD\iv\vy \|_{\infty} < e^t \eps.
 \]
  We will show that the error vector in the \hkre steps, $\tnps - \vy$, satisfies $\| \mD\iv \tnps - \mD\iv\vy \|_{\infty} < e^t \eps/2$.

The following lemma expresses the error vector, $\tnps - \vy$, as a weighted sum of the residual blocks $\vr_k$ in the linear system \eqref{largelinsys}, and shows that the polynomials $\psi_k(t)$ are the weights.

\begin{lemma}\label{lem:psiweights}\label{lem:psi}
Let $\psi_k(t)$ be defined as in Section \ref{sec:alg}. Then in the notation of Section \ref{sec:linsys}, we can express the error vector of \hkre  in terms of the residual blocks $\vr_k$ as follows
\begin{equation}\label{err-res-sum}
\tnps - \vy = \sum_{k=0}^N \psi_k(t\mP) \vr_k
\end{equation}
\end{lemma}
% MAIN LEMMA PROOF
\begin{proof}
Consider \eqref{largelinsys}. Recall that $\vv = [ \vv_0; \vv_1; \cdots; \vv_N]$ and let $\mS$ be the $(N+1)\times(N+1)$ matrix of 0s with first subdiagonal equal to $[\frac{1}{1},$ $\frac{1}{2},$ $\cdots,$ $\frac{1}{N}]$. Then we can rewrite this linear system more conveniently as 
\begin{equation}\label{linsys}
(\mI - \mS \kron (t\mP)) \vv = \ve_1\kron\vs.
\end{equation}
Let $\vv_k$ be the true solution vectors that the $\hvv_k$ are approximating in \eqref{linsys}. We showed in Section \ref{sec:linsys} that the error $\tnps - \vy$ is in fact the sum of the errors $\vv_k - \hvv_k$. Now we will express $\tnps - \vy$ in terms of the residual partitions, i.e. $\vr_k$.

At any given step we have
$\vr = \ve_1\kron\vs - (\mI - \mS \kron (t \mP)) \hvv$,
so pre-multiplying by $(\mI - \mS \kron (t \mP))\iv$ yields
$ (\mI - \mS \kron (t \mP))\iv \vr = \vv - \hvv,$
because $(\mI - \mS \kron (t \mP)) \vv = \ve_1 \kron \vs$ exactly, by definition of $\vv$. Note that $ (\mI - \mS \kron (t \mP))\iv \vr = \vv - \hvv$
is the error vector for the linear system \eqref{linsys}. From this, an explicit computation of the inverse (see~\cite{Kloster-2013-nexpokit} for details) yields
\begin{equation}\label{errorblocks}
\bmat{\vv_0 - \hvv_0 \\ \vdots \\ \vv_N - \hvv_N } =  \lp \sum_{k=0}^N \mS^k \kron (t \mP)^k \rp \fullvec{\vr}.
\end{equation}
For our purposes, the full block vectors $\vv,\hvv,\vr$ and their individual partitions are unimportant: we want only their sum, because $\tnps - \vy = \sum_{k=0}^N (\vv_k - \hvv_k)$, as previously discussed.

Next we use \eqref{errorblocks} to express $\sum_{k=0}^N (\vv_k - \hvv_k)$, and hence \[\tnps - \vy ,\] in terms of the residual blocks $\vr_k$.
We accomplish this by examining the coefficients of an arbitrary block $\vr_k$ in \eqref{errorblocks}, which in turn requires analyzing the powers of $(\mS\kron t\mP)$.

Fix a residual block $\vr_{j-1}$ and consider the product with a single term $(\mS^k\kron (t\mP)^k)$. Since $\vr_{j-1}$ is in block-row $j$ of $\vr$, it multiplies with only the block-column $j$ of each term $(\mS^k\kron (t\mP)^k)$, so we want to know what the blocks in block-column $j$ of $(\mS^k\kron (t\mP)^k)$ look like.

From~\cite{Kloster-2013-nexpokit} we know that
\[
\mS^m\ve_j =
 \begin{cases} 
      \frac{(j-1)!}{(j-1+m)!}\ve_{j+m}   & \text{if } 0\leq m \leq N+1-j \\
      0  & \text{otherwise.} \\
  \end{cases}
\] 
This means that block-column $j$ of $(\mS^m \kron(t\mP)^m)$ contains only a single nonzero block, $\frac{(j-1)!}{(j-1+m)!}  t^m \mP^m$, for all $0 \leq m \leq N+1-j$. Hence, summing the $n \times n$ blocks in block-column $j$ of each power $(\mS^m \kron (t\mP)^m)$ yields
\begin{equation}\label{rescoeff}
\sum_{m=0}^{N+1-j} \frac{(j-1)! t^m }{(j-1+m)!} \mP^m
\end{equation}
as the matrix coefficient of $\vr_{j-1}$ in the right-hand side expression of
$\sum_{k=0}^N (\vv_k - \hvv_k) = (\ve^T \kron \mI)(\sum_{m=0}^N \mS^m \kron (t \mP)^m ) \vr$.

  Substituting $k = j-1$ on the right-hand side of \eqref{rescoeff} yields
  \begin{align*}
  \sum_{k=0}^N (\vv_k - \hvv_k) &=  (\ve^T \kron \mI)\lp\sum_{m=0}^N \mS^m \kron (t \mP)^m \rp \vr \\
  &= \sum_{k=0}^N \left( \sum_{m=0}^{N-k} \frac{k! t^m }{(k+m)!} \mP^m \right) \vr_k \\
  &= \sum_{k=0}^N   \psi_k(t\mP) \vr_k,
\end{align*} as desired.
\end{proof}

Now that we've shown $\tnps - \vy = \sum_{k=0}^N   \psi_k(t\mP) \vr_k$ we  can upperbound $\| \mD\iv \tnps - \mD\iv\vy \|_{\infty}$. To do this, we first show that 
\begin{equation}\label{psiclaim}
\| \mD\iv\psi_k(t\mP) \vr_k \|_{\infty} \leq \psi_k(t)  \|\mD\iv\vr_k\|_{\infty}.
\end{equation}

To prove this claim, recall that $\mP = \mA \mD\iv$. Since $\psi_k(t)$ is a polynomial, we have $\mD\iv\psi_k(t\mP) = \mD\iv \psi_k(t\mA\mD\iv) = \psi_k(t \mD\iv \mA) \mD\iv$, which equals $\psi_k(t\mP^T)\mD\iv$.

Taking the infinity norm and applying the triangle inequality to the definition of $\psi_k(t)$ gives us
\begin{align*}
\|\mD\iv\psi_k(t\mP) \vr_k\|_{\infty} \leq \sum_{m=0}^{N-k} \frac{t^m k!}{(m+k)!} \| (\mP^T)^m (\mD\iv\vr_k) \|_{\infty} .
\end{align*}
Then we have
$
\|(\mP^T)^m (\mD\iv\vr_k) \|_{\infty} \leq \|(\mP^T)^m\|_{\infty} \|(\mD\iv\vr_k) \|_{\infty}
$
which equals $ \|(\mD\iv\vr_k) \|_{\infty}$ because $\mP^T$ is row-stochastic. Returning to our original objective, we then have
\begin{align*}
\|\mD\iv\psi_k(t\mP) \vr_k\|_{\infty} &\leq \sum_{m=0}^{N-k} \frac{t^m k!}{(m+k)!} \|(\mD\iv\vr_k) \|_{\infty}  \\
&= \psi_k(t) \| \mD\iv \vr_k \|_{\infty},
\end{align*}
proving the claim in \eqref{psiclaim}.

This allows us to continue:
\begin{align*}
 \tnps - \vy  &= \sum_{k=0}^N \psi_k(t\mP)\vr_k & \text{so}\\
 \| \mD\iv \tnps - \mD\iv\vy\|_{\infty} &\leq \sum_{k=0}^N\|\mD\iv \psi_k(t\mP)\vr_k\|_{\infty}\\
&\leq \sum_{k=0}^N \psi_k(t) \| \mD \iv\vr_k\|_{\infty}.
\end{align*}

The stopping criterion for \hkre requires that every entry of the residual satisfies $r(i,j) < \frac{ e^t \eps d_i }{2N\psi_j(t)}$, which is equivalent to satisfying $\frac{r(i,j)}{d_i} < \frac{e^t \eps}{2N \psi_j(t)}$. This condition guarantees that $\|\mD\iv \vr_k\|_{\infty} < \frac{e^t \eps}{2N\psi_k(t)}$. Thus we have
\begin{align*}
 \| \mD\iv\tnps - \mD\iv\vy\|_{\infty} &\leq \sum_{k=0}^N \psi_k(t) \| \mD \iv\vr_k\|_{\infty} \\
 & \leq \sum_{k=0}^N \lp \psi_k(t)\frac{e^t \eps}{2N\psi_k(t)} \rp
 \end{align*}
which is bounded above by $e^t \eps/2$. Finally, we have that \[ \| \mD\iv\tnps - \mD\iv\vy\|_{\infty} < e^t \eps/2 \] implies $\| \mD\iv ( \expof{-t(\mI-\mP)}\vs - \vx )\|_{\infty} < \eps/2$, completing our proof of the first part of Theorem \ref{thm:work}.
\vfill\eject
\paragraph{Bounding work} It remains to bound the work required to perform steps of \hkre until the stopping criterion is achieved.

Because the stopping criterion requires each residual entry to satisfy
\[
\frac{r(i,j)}{d_i} < \frac{e^t\eps}{2N\psi_j(t)}
\]
we know that each step of \hkre must operate on an entry $r(i,j)$ larger than this threshold. That is, each step we relax an entry of $\vr$ satisfying $\frac{r(i,j)}{d_i} \geq \frac{e^t\eps}{2N\psi_j(t)}$.

Consider the solution vector $\vy \approx \tnps$ and note that each entry of $\vy$ is really a sum of values that we have deleted from $\vr$. But $\vr$ always contains only nonnegative values: the seed vector $\vs$ is nonnegative, and each step involves setting entry $r(i,j) =0$ and adding a scaled column of $\mP$, which is nonnegative, to $\vr$. Hence, $\|\vy \|_1 $ equals the sum of the $r(i,j)$ added to $\vy$, $\sum_{l=1}^T r(i_l,j_l) = \|\vy \|_1$.

Finally, since the entries of $\vy$ are nondecreasing (because they only change if we add positive values $r(i,j)$ to them), we know that $\|\vy\|_1 \leq \|\tnps\|_1 = \psi_0(t) \leq e^t$. This implies, then, that
\[
\sum_{l=1}^T r(i_l,j_l) \leq e^t.
\]
Using the fact that the values of $r(i,j)$ that we relax must satisfy $ \frac{r(i,j)}{d_i} \geq \frac{e^t\eps}{2N\psi_j(t)}$ we have that
\[
\sum_{l=1}^T \frac{d_{i_l} e^t \eps}{2N \psi_{j_l}(t)} \leq e^t.
\]
Simplifying yields
\[
\sum_{l=1}^T \frac{d_{i_l}}{\psi_{j_l}(t)} \leq \frac{2N}{\eps}.
\]
By Lemma \ref{lem:psi} we know $\psi_k(t) \leq \psi_1(t)$ for each $k\geq 1$, so we can lowerbound
$\sum_{l=1}^T \frac{d_{i_l}}{\psi_1(t)} \leq \sum_{l=1}^T \frac{d_{i_l}}{\psi_{j_l}(t)} \leq \frac{2N}{\eps}$, giving us
\[
\sum_{l=1}^T d_{i_l} \leq \frac{2N\psi_1(t)}{\eps}.
\]

Finally, note that the dominating suboperation in each step of \hkre consists of relaxing $r(i,j)$ and spreading this ``heat kernel rank" to the neighbors of node $i$ in block $\vr_{j+1}$. Since node $i$ has $d_i$ neighbors, this step consists of $d_i$ adds, and so the work performed by \hkre is
$\sum_{l=1}^T d_{i_l}$, which is exactly the quantity we bounded above.

%
% ACM needs 'a single self-contained file'!
\end{document}